\documentclass[preprint,11pt,authoryear]{elsarticle}

\usepackage{lipsum}

\usepackage{amssymb}
\usepackage{amsmath}

\newtheorem{theorem}{Theorem}
\newtheorem{example}{Example}%
\newtheorem{remark}{Remark}%
\newtheorem{corollary}{Corollary}
\newtheorem{algorithm}{Algorithm}
\newenvironment{proof}{\paragraph{Proof:}}{\hfill$\square$}
\usepackage{siunitx}
\usepackage{algorithmic}

\begin{document}

\begin{frontmatter}



\title{Some general results on risk budgeting portfolios} 


 \author[label1]{Claudia Fassino} \author[label2]{Pierpaolo Uberti}
 \affiliation[label1]{organization={University of Genoa},
             addressline={Via Dodecaneso 35},
             city={Genoa},
             postcode={16146},
             country={Italy}}

 \affiliation[label2]{organization={University of Milano-Bicocca},
             addressline={Via Bicocca degli Arcimboldi 8},
             city={Milan},
             postcode={20126},
             country={Italy}}



\begin{abstract}
Given a reference risk measure, the risk budgeting is the portfolio where each asset contributes a predetermined amount to the total risk. We propose a novel approach—alternative to the ones proposed in the literature—for the calculation of the risk budgeting portfolio. This different perspective on the problem has several interesting consequences. For the calculation of the portfolio, we define a Cauchy sequence within the simplex of $\mathbb R^n$, whose limit corresponds to the risk budgeting portfolio. This construction allows for the straightforward implementation of an efficient algorithm, avoiding the need to solve auxiliary, equivalent optimization problems, which may be computationally challenging and hard to interpret in the decision theory context. We compare our algorithm with the standard optimization-based methods proposed in the literature. From a theoretical point of view, starting from the Cauchy sequence, we define a function for which the risk budgeting portfolio is a fixed point. Therefore, sufficient conditions for the existence and uniqueness of the fixed point can be used. The methodology is developed for general risk measures and implemented in detail in the case of standard deviation.
\end{abstract}




\end{frontmatter}

\section{Introduction}

Risk parity refers to an investment strategy in which each asset contributes equally to the total portfolio risk. In its original formulation, the risk parity portfolio was introduced using the standard deviation of returns as the reference risk measure. The idea of allocating the wealth by equally distributing risk exposures across assets is simple and intuitive. This strategy has attracted considerable attention from both academic researchers and practitioners, as it offers a useful alternative to classical allocation models -primarily inspired by the Markowitz framework (see \cite{M})- which struggled to address the challenges revealed by the 2008 global financial crisis.

Standard optimization-based allocation techniques suffer from several well-known shortcomings. These approaches are subject to model instability, as pointed out in \cite{kan2007optimal}, and to numerical instability: even infinitesimal variations in the model’s inputs can lead to a large changes in the optimal allocation, see \cite{best1991sensitivity}, causing the in-sample frontier to be a biased estimator of the real frontier, see \cite{kan2008distribution}. Furthermore, optimization models often tend to concentrate the allocation in a small number of assets, see \cite{yanushevsky2015approach}. These drawbacks are among the principal reasons why the equally weighted portfolio and other heuristic allocation strategies frequently outperforms optimization-based approaches in out-of-sample experiments (see, among others, \cite{demiguel2009optimal}, \cite{yuan2024naive} and \cite{gelmini2024equally}). Risk parity portfolios are attractive because they do not exhibit these limitations. They are only slightly sensitive to parameters uncertainty and yield numerical stable allocations that are not concentrated in a few assets. Thus, the practical implementation is straightforward and transaction costs have negligible impact on performance. Over the past 15 years, numerous authors have studied the risk parity portfolio using different risk measures. Notable examples include \cite{hallerbach1999decomposing} for Value at Risk, \cite{mausser2018long} for Expected Shortfall, \cite{ararat2024mad} for Mean Absolute Deviation, and \cite{bellini2021risk} for expectiles.

Risk budgeting is a natural generalization of risk parity: for a given risk measure, the investor decides the risk exposure contribution of each asset and, subsequently, computes the allocation satisfying the predetermined risk exposures. Despite its simple definition, no closed-form solution exists for the risk budgeting portfolio, making numerical methods necessary for its calculation. One central issue concerning risk budgeting portfolios is the choice of the reference risk measure. Ideally, an investor can select any risk measure and subsequently compute the corresponding risk budgeting portfolio. To this end, it is necessary to calculate the marginal risk contribution of each asset to the total risk, given the chosen risk measure. A very recent paper provides general results on this topic. In \cite{cetingoz2024risk}, the existence and uniqueness of the risk budgeting portfolio for general risk measures are investigated. Further research, limited to the class of so-called coherent risk measures (see \cite{artzner1999coherent}), focuses on approximating the risk budgeting portfolio using simulated data \cite{da2023risk}. Both approaches calculate the risk budgeting portfolio by solving an auxiliary optimization problem, which generalizes the optimization framework originally proposed in \cite{maillard2010properties} for the risk parity portfolio when standard deviation serves as the reference risk measure. While \cite{maillard2010properties} is a milestone in this context, it lacks a formal mathematical proof of the existence and uniqueness of the risk parity portfolio. The first such proof for the risk parity portfolio under the standard deviation measure can be found in \cite{spinu}. Let us underline that both the results on the existence and uniqueness of the risk budgeting portfolio and its calculation strongly rely on the definition and solution of an auxiliary optimization problem. This approach, although standard, has some known issues. The optimization problem does not have an immediate economic interpretation, so it is difficult to explain why an investor should allocate her/his wealth by solving that problem. From an empirical point of view, the computation of the risk budget allocation suffers from issues related to numerical optimization. From a theoretical point of view, the results on the uniqueness and existence of the risk budgeting portfolio entirely rely on the convexity of the optimization problem.

We study the risk budgeting portfolio from a novel and alternative perspective. As often happens in mathematics, a different perspective on a given problem can unveil interesting new aspects. We start by defining a Cauchy sequence within the simplex of $\mathbb{R}^N$. Since the simplex with the Euclidean norm forms a complete metric space, every Cauchy sequence converges to a point within this space. By construction, the limit of the sequence corresponds to the risk budgeting portfolio. This approach immediately leads to an algorithm for computing the risk budgeting portfolio. Through numerical simulations, we show that the algorithm is both fast and efficient compared to standard optimization-based approaches proposed in the literature. From a theoretical standpoint, we show that the Cauchy sequence enables the definition of a function on the simplex for which the risk budgeting portfolio is a fixed point by construction. Therefore, conditions for the existence and uniqueness of the fixed point can be used to investigate the existence and uniqueness of the risk budgeting portfolio. Our approach is introduced in a general framework, independently of the choice of a reference risk measure. Subsequently, we exploit the method to the case of the standard deviation and detail its implementation. Our contribution to the literature is multifaceted. Theoretically, we provide sufficient conditions for the existence and uniqueness of the risk budgeting portfolio within the simplex. These conditions depend on the functional form of the marginal risk contribution of each asset to the total risk, relative to the chosen risk measure. Empirically, we derive an efficient algorithm for calculating the risk budgeting portfolio, avoiding the need to solve equivalent auxiliary optimization problems.

The paper is organized as follows: Section \ref{ass} presents the main idea, the results for a general risk measure, and the structure of the algorithm. Section \ref{stdev} applies the methodology to the case where the reference risk measure is the standard deviation. Section \ref{num} is devoted to the numerical results, while the conclusions are presented in Section \ref{sec:conclusions}.

\section{The general idea for an arbitrary risk measure}
\label{ass}

In this section, we illustrate the general idea behind our approach. Its practical implementation requires selecting a risk measure and calculating the marginal risk contribution of each asset with respect to that measure. We do not focus on this aspect here and refer to the papers cited in the introduction, which study the risk parity portfolio for specific risk measures.\footnote{The papers cited in the introduction focus on the risk parity portfolio, which is a special case of risk budgeting. Nevertheless, the problem of determining the marginal risk contribution of an asset with respect to a given risk measure is common to both risk parity and risk budgeting.}
This section is divided into four subsections. The first subsection reviews the basic assumptions and notation. The second subsection is dedicated to constructing the Cauchy sequence that converges to the risk budgeting portfolio. The third subsection presents sufficient conditions for the existence and uniqueness of the risk budgeting portfolio via a fixed-point approach. Finally, the last subsection provides the pseudo-code for the practical computation of the risk budgeting portfolio.

\subsection{Notation and basic assumptions}
We assume that $N$ risky assets are available on a market. A portfolio is identified by the column vector $\mathbf{x} \in \mathbb{R}^{N}$ where the $i$-th entry $x_i$, with $i = 1,\ldots,N$, is the share of the individual wealth invested in asset $i$. The analysis is restricted to long-only portfolios: no short positions, represented by negative weights, are allowed. A portfolio $\mathbf{x}$ is an element of the simplex $S$ in $\mathbb{R}^{N}$, 
$$S = \left\{\mathbf{x} = [x_1,\ldots, x_N]^t \in \mathbb R^N : \sum_{i=1}^{N} x_i=1, \; \textrm{ and } x_i \geq 0, \; \forall i =1,\ldots,N \right\}.$$ 
The superscript $t$ represents the transpose operation. If nothing different is specified, we always refer to column vectors. The symbols $\mathbf{1}, \mathbf{0} \in \mathbb{R}^{N}$ identify the $N$-dimensional column vectors with unitary and null entries respectively; $\textrm{diag}(\mathbf{x})$, with $\mathbf{x} \in \mathbb{R}^{N}$, is the diagonal matrix of order $N$ where the diagonal elements are the entries of $\mathbf{x}$.

Given a standard probability space $(\Omega, F, \mathbb P)$, $\rho:L^0(\Omega, \mathbb R) \rightarrow \mathbb R$ is the risk measure, where $L^0(\Omega, \mathbb R)$ is the set of $\mathbb R$-valued random variables, the measurable functions defined on $\Omega$ with real values. In the present framework, the random variables are the returns of the assets. The risk contribution of asset $i$ to the total risk with respect to the risk measure $\rho$ is $\textsl{RC}^{\rho}_i(\mathbf{x})$, for $\mathbf{x} \in S$. The vector $\textsl{RC}^{\rho}(\mathbf{x}) \in \mathbb{R}^{N}$ contains the risk contributions of the $N$ assets. 

The risk budget is identified by $\mathbf{x}^b \in S$. This vector is given by the investor and represents the objective in terms of risk allocation across the $N$ assets. Consequently, ${x}^b_i$ is the contribution of asset $i$ to the portfolio risk. Obviously, if ${x}^b_i = \frac{1}{N}$ for $i = 1, \ldots, N$, the risk budgeting portfolio boils down to the risk parity portfolio. The vector $\mathbf{x}^* \in S$ that provides the desired allocation in terms of risk contributions, i.e. $\textsl{RC}^{\rho}(\mathbf{x}^*)/\boldsymbol{1}^t\textsl{RC}^{\rho}(\mathbf{x}) = \mathbf{x}^b$, is the risk budgeting portfolio.\footnote{In general, $\boldsymbol{1}^t\textsl{RC}^{\rho}(\mathbf{x}) \neq 1$. Then, to obtain a vector in the simplex it is necessary to normalize dividing by the sum of the elements of $\textsl{RC}^{\rho}$.} As already underlined in the introduction, the existence and uniqueness of $\mathbf{x}^* \in S$ for the so-called risk budgeting compatible risk measures has been recently studied in \cite{cetingoz2024risk}. 

\subsection{The sequence that converges to the risk budgeting portfolio}\label{fix_gen}

Limited to the present subsection, we assume that the risk budgeting portfolio exists and is unique in $S$ for a given risk measure $\rho$. Each choice of the reference risk measure $\rho$ corresponds to a specific $\textsl{RC}^{\rho}(\mathbf{x})$. When no confusion arises, we omit $\mathbf{x}$ and simply write $\textsl{RC}^{\rho}$. The formalization of the general idea behind the proposal does not require to specify the functional form of $\textsl{RC}^{\rho}$. We will exploit the calculations in detail in section \ref{stdev} when applying the approach in the case of the standard deviation. Let us define the functions $\Delta: \mathbb R^N \rightarrow \ \mathbb R^N$ and $G: \mathbb R^N \rightarrow \mathbb R^N$ as 
\begin{equation}
\Delta (\mathbf{x}) = \textsl{RC}^{\rho}(\mathbf{x})  - \mathbf{x}^b \textbf{1}^t \textsl{RC}^{\rho}(\mathbf{x})
\end{equation}
\begin{equation}\label{functionG}
G(\mathbf{x})=\mathbf{x} - k(\mathbf{x})\Delta (\mathbf{x} )\ .
\end{equation}

Given an initial point $\mathbf{x}_1 \in S$, we define the sequences $\{\mathbf{x}_n\}_{n=1,2,\dots}$ and $\{\Delta(\mathbf x_n)\}_{n=1,2,\dots}$ as 
\begin{eqnarray} \label{succ}
\left \{ \begin{array}{lcl}
\Delta( \mathbf{x}_n )&=& \textsl{RC}^{\rho}(\mathbf{x}_n)  - \mathbf{x}^b \textbf{1}^t \textsl{RC}^{\rho}(\mathbf{x}_n) \\ \\ 
\mathbf{x}_{n+1} &=& \mathbf{x}_n + k(\mathbf{x}_n) \Delta( \mathbf{x}_n) 
\end{array} \right . \ ,
\end{eqnarray}
where the function $k(\mathbf{x})$ is opportunely chosen as described in the following. The functions $\Delta(\mathbf{x})$ and $k(\mathbf{x})$ are assumed to be continuous on their domain. This hypothesis is fundamental for proving the convergence of $\{\mathbf{x}_n\}_{n=1,2,\dots}$. The continuity of $\Delta$ requires that the risk contribution $\textsl{RC}^{\rho}$ is continuous. Consequently, only risk measures characterized by continuous risk contributions are compatible with the risk budgeting problem.\footnote{Given the arbitrariness of $\mathbf{x}^b \in S$, it is evident that continuity of the risk contribution is necessary for the existence of the risk budgeting portfolio. We underline that the existence and uniqueness of the risk budgeting portfolio for a given risk measure do not depend on the computational framework used to determine it. Nevertheless, different approaches to the problem may identify distinct classes of risk measures compatible with risk budgeting.} 
\begin{remark}
We note that the condition $\mathbf{x}_1 \in S$ does not guarantee that the entire sequence $\{\mathbf{x}_n\}_{n=1,2,\dots}$ remains in $S$. It holds that $\textbf{1}^t \mathbf{x}_n = 1$ for all $n \in \mathbb{N}$; in other words, each vector in the sequence sums to $1$,  but some entries may become negative. To ensure that the sequence $\{\mathbf{x}_n\}$ remains within the simplex $S$, it is sufficient to bound the norm of $k(\mathbf{x}_n)$. We will explicitly show how to restrict the sequence to the simplex in section \ref{stdev}.  
\end{remark}
Assuming the existence of a continuous function $k$, we prove that $\{\mathbf{x}_n\}$ is a Cauchy sequence. In the proof of the theorem, we use the euclidean norm, $\|\cdot\|_2$. The choice of the norm is not fundamental.\footnote{Although the result does not depend on the specific choice of the norm and the corresponding distance, $\mathbb R^N$ with the euclidean norm is a complete metric space. Completeness ensures that every Cauchy sequence converges to a point within the space.} 

\begin{theorem}\label{main}
Let $k(\mathbf{x}_n)$ such that $\|\Delta( \mathbf{x}_n) \|_2 \geq L \|\Delta (\mathbf{x}_{n+1}) \|_2$, with $0<L<1$, and $|k(\mathbf{x}_n)| < k$ for $n \in \mathbb N$, (i.e. the sequence $\{k(\mathbf{x}_n)\}_{n=1,2,\dots} $ is bounded), then 
$$\lim_{n\rightarrow \infty} \Delta (\mathbf{x}_n)=\mathbf{0} $$
and $\{\mathbf{x}_n\}$ is a Cauchy sequence.
\end{theorem}
\begin{proof} {First of all, by assumption, it exists $0<L<1$ such that $\|\Delta(\mathbf x_{n+1})\|_2 \le L \|\Delta(\mathbf x_n)\|_2$, then it holds that
\begin{equation}\label{Delta_upper_bound}
\|\Delta(\mathbf x_{n}) \|_2 \le L \|\Delta(\mathbf x_{n-1} )\|_2 \le \dots\le  L^{i} \|\Delta(\mathbf x_{n-i} )\|_2\le\dots\le   L^n \|\Delta(\mathbf x_0)\|_2 
\end{equation}
and so, since $L<1$,
$$\lim_ {n\rightarrow \infty} \|\Delta (\mathbf{x}_n)\|_2=0$$
or, equivalently, 
$$\lim _{n\rightarrow \infty} \Delta (\mathbf{x}_n)=\mathbf{0}\ .$$
A Cauchy sequence is such that
$$\forall \epsilon > 0 \quad \exists \overline{n}: \| \mathbf{x}_p -  \mathbf{x}_q\|_2 < \epsilon, \quad \forall p,q > \overline{n} \ . $$
\vspace{5px}
Without loss of generality, $q = p + j$, with $j \in \mathbb{N}$. Then
\begin{eqnarray*}
\|   \mathbf{x}_p-  \mathbf{x}_q\|_2 &=&\|  \mathbf{x}_p- \mathbf  x_{p+j} \|_2 =\|  \mathbf{x}_p - \mathbf x_{p+1}+ \mathbf x_{p+1} -\mathbf  x_{p+j} \|_2 \\
&&\le \|  \mathbf{x}_p- \mathbf x_{p+1}\|_2 + \|\mathbf x_{p+1} - \mathbf x_{p+j} \|_2 \\
&&\le \|  \mathbf{x}_p- \mathbf x_{p+1}\|_2 + \|\mathbf x_{p+1} - \mathbf x_{p+2} \|_2 + \|\mathbf x_{p+2} - \mathbf x_{p+j} \|_2 \\
&&\le  \dots \le  \sum_{i = 0}^{j-1} \|\mathbf x_{p+i}-\mathbf x_{p+i+1}\|_2 = \sum_{i = 0}^{j-1} \|k(\mathbf{x}_{p+i})\Delta(\mathbf  x_{p+i})\|_2 \\
&& \le k \sum_{i = 0}^{j-1} \|\Delta (\mathbf x_{p+i})\|_2 \ ,
\end{eqnarray*}
since we suppose that the sequence $ \{k(\mathbf{x}_n)\}$ is bounded.
From~\eqref{Delta_upper_bound} with $n=p+i$, it follows $\|\Delta (\mathbf x_{p+i})\|_2 \le L^i \|\Delta (\mathbf x_{p})\|_2$ and so
\begin{eqnarray*}
\|  \mathbf{x}_p -   \mathbf{x}_q\|_2 && \le k \sum_{i = 0}^{j-1} \|\Delta(\mathbf x_{p+i})\|_2 \le k \sum_{i = 0}^{j-1} L^i \|\Delta (  \mathbf{x}_p)\|_2 \\
&&< k\|\Delta (  \mathbf{x}_p)\|_2 \sum_{i = 0}^{\infty }  L^i= k\|\Delta(  \mathbf{x}_p)\|_2 \frac 1 {1-L} \ ,
\end{eqnarray*}
because $\lim_{n \rightarrow \infty}\sum_{i = 0}^n L^i = \frac{1}{1-L}$.
Since the limit of the sequence $\{\|\Delta( \mathbf x_n)\|_2\}$ is $0$, given $\epsilon$ and $\gamma=\epsilon\frac{1-L}k$, it exists $\overline n$ such that, if $p> \overline n$ then $\|\Delta  ( \mathbf{x}_p)\|_2< \gamma$. \\
Concluding, for each $\epsilon$ it exists $\overline n$ such that for $p> \overline n$ and, consequently, $q>p>\overline n$, the following holds 
\begin{eqnarray*}
\|   \mathbf{x}_p -   \mathbf{x}_q\|_2 \le k\|\Delta   (\mathbf{x}_p)\|_2 \frac 1 {1-L} < k \gamma  \frac 1 {1-L} =\epsilon \ .
\end{eqnarray*}}
\end{proof}

\begin{remark}
Theorem \ref{main} assumes the existence of a function $k(\mathbf x_n)$ such that $\|\Delta (\mathbf{x}_n) \|_2 \leq L \|\Delta (\mathbf{x}_{n+1}) \|_2$ holds for all $n \in \mathbb N$. In section \ref{stdev}, we will show how to calculate $k(\mathbf x_n)$ in the specific case where the reference risk measure is the standard deviation. In practice, given a constant $0<L<1$, it is sufficient to solve the inequality above to determine $k(\mathbf x_n)$. The existence of at least one real solution depends on the value of $L$. Within the present framework, $L$ plays a role similar to that of the Lipschitz constant in the context of contraction mappings, see \cite{banach1922operations}. However, deriving a suitable value of $L$ for all $\mathbf{x} \in S$ is challenging, as it depends on the geometry of the problem. Furthermore, since $k(\mathbf x_n)$ depends on the functional form $\textsl{RC}^{\rho}$, it cannot be computed without reference to a specific risk measure. In the application, we arbitrarily choose a value $0<L<1$. If $L$ is sufficiently close to $1$, the continuity of $\Delta$ guarantees the existence of at least one real solution to the inequality, although the convergence rate of the sequence is slow in this case. Conversely, when $L$ is closer to zero, the convergence rate is potentially faster, but a real solution to the inequality may not exist. For these reasons, the choice of $L$ is crucial. We will provide guidance on how to proceed when the inequality admits no real solutions in $\mathbb R$ in section \ref{stdev}.
\end{remark}

Since $\mathbb R^N$ with the euclidean distance is a complete metric space, each Cauchy sequence has a finite limit within the space. Therefore, the sequence $\{ \mathbf{x}_n\}$ has a finite limit.  
\begin{theorem}\label{fix}
Under the assumptions of theorem~\ref{Delta_upper_bound}, if $\lim_{n\rightarrow \infty} \mathbf{x}_n=\widetilde{ \mathbf{x}} $ for $\widetilde{\mathbf{x}}\in \mathbb R^N$, then 
$$ \Delta(\widetilde {\mathbf{x}})=\mathbf{0} $$
that is 
$$\textsl{RC}^{\rho}(\widetilde{ \mathbf{x}} )  = \mathbf{x}^b \textbf{1}^t \textsl{RC}^{\rho}(\widetilde{ \mathbf{x}} ) \ .$$
\end{theorem}
\begin{proof}
Thanks to the assumptions on the existence of $k(\mathbf{x}_n)$ (with opportune properties), see theorem \ref{main},   $\lim_{n\rightarrow \infty} \Delta( \mathbf{x}_n )=\mathbf{0}$. \\
The function $\Delta$ is continuous and $\lim_{n\rightarrow \infty} \mathbf{x}_n=\widetilde{ \mathbf{x}} $, then it holds that
$$ \lim_{n\rightarrow \infty} \Delta( \mathbf{x}_n)= \Delta (\widetilde{ \mathbf{x}} ) =  \textsl{RC}^{\rho}(\widetilde{ \mathbf{x}} )  - \mathbf{x}^b \textbf{1}^t \textsl{RC}^{\rho}(\widetilde{ \mathbf{x}} )  \ .$$
Concluding,  
$$\Delta(\widetilde {\mathbf{x}})=\mathbf{0}  \quad \textrm {and} \quad  \textsl{RC}^{\rho}(\widetilde{ \mathbf{x}} )  - \mathbf{x}^b \textbf{1}^t \textsl{RC}^{\rho}(\widetilde{ \mathbf{x}} )  =\mathbf 0 $$
from which the thesis of the theorem.
\end{proof}
Theorem \ref{fix} shows that the limit of the Cauchy sequence is the risk budgeting portfolio, provided that $\widetilde{ \mathbf{x}} \in S$.

\subsection{Existence and uniqueness of the fixed point in the simplex}\label{exuniq}

This section presents sufficient conditions for the existence and uniqueness of the risk budgeting portfolio in $S$.
By Theorem \ref{fix}, $\widetilde{ \mathbf{x}}$ is a fixed point of the function $G$, where $k$ is the continuous function defined in the previous subsection. The result follows directly from the observation that $\Delta (\widetilde{ \mathbf{x}} ) = \mathbf{0}$ implies $G(\widetilde{ \mathbf{x}} )=\widetilde{ \mathbf{x}} $, since $k$ is a bounded function.
Moreover, we observe that the risk budgeting portfolio $\mathbf{x}^*$ is also a fixed point of the function $G$ in $S$, since $ \mathbf{x}^b = RC^\rho(\mathbf{x}^*)/\boldsymbol{1}^t RC^\rho(\mathbf{x}^*)$ implies
\begin{eqnarray}\label{fixed_risk}
 G(\mathbf x^*)&=& \mathbf x^* - k(\mathbf x^*)\left (RC^\rho(\mathbf{x}^*) - \mathbf x^b \mathbf 1^t  RC^\rho(\mathbf{x}^*)\right ) =   \mathbf x^* -  k(\mathbf x^*)\left (RC^\rho(\mathbf{x}^*) - RC^\rho(\mathbf{x}^*) \right )= \mathbf  x^* \ . \nonumber
 \end{eqnarray}
If the limit of $\{\mathbf{x}_n\}$ belongs to $S$, and if the function $G$ admits a unique fixed point in $S$, then the limit of $\{\mathbf x_n\}$ is the risk budgeting portfolio. While in the majority of the empirical applications that we have performed -using the standard deviation as the reference risk measure- $G$ behaves as a contraction on $S \subset \mathbb R^N$, it is possible to build examples where $G$ is not a contraction. Consequently, the uniqueness of the fixed point is not guaranteed a priori. To prove that $G$ admits a unique fixed point in $S$, we refer to Kellogg's theorem, see \cite{kellogg1976uniqueness}. Let $X$ be a real Banach space with a bounded convex open subset $D$, and let $F: \overline D \rightarrow \overline D$ be a continuous function, which is also assumed to be compact if $X$ is infinite dimensional. Under the assumption that $F$ is differentiable, there is a condition which guarantees the result. If the only relevant case is the finite dimensional, the compactness hypothesis can be omitted.

\begin{theorem}\label{kellogg} [\textbf{Kellogg theorem}] 
 Let  $F: \overline D \rightarrow \overline D$ be a compact continuous map which is continuously Frechet differentiable on $D$. Suppose that 
 \begin{itemize}
 \item for each $\mathbf{x} \in D$, $1$ is not an eigenvalue of $F^\prime$, and 
 \item for each $\mathbf{x} \in \partial D$, $\mathbf{x} \neq F(\mathbf{x})$. 
 \end{itemize}
 Then $F$ has a unique fixed point.
 \end{theorem} 

Without loss of generality, we assume that the dimension of the reference space -the number of assets $N$- is finite, and we consider the interior of $S$, denoted by
$$\mathring{S} = \left\{\mathbf{x} = [x_1;\ldots;x_N]^t \in \mathbb R^N : \sum_{i=1}^{N} x_i=1, \; \textrm{ and } x_i > 0, \; \forall i =1,\ldots,N \right\}\ . $$
Equation~\eqref{fixed_risk} establishes the existence of the fixed point of $G$ in $S$. However, the economic interpretation is particularly important: the risk budgeting portfolio is characterized by a strictly positive allocation in each of the $N$ assets. It never concentrates the allocation in a small number of large positions, and therefore $x^* \in \mathring{S}$.

The role of $k(\mathbf{x})$ is fundamental in ensuring that the function $G$ satisfies the necessary conditions on $S$. If the functions $k$ and $\Delta$ are continuous for all $\mathbf x \in S$, then $G$ is also continuous, and $k$ is bounded on the compact set $S$. Furthermore, if $k$ is appropriately chosen, the following holds:
\begin{enumerate}
 \item $G(\mathbf{x}) \in S$ for all $\mathbf x \in S$;  \label{1}
 \item $G$ satisfies the assumptions of Kellogg's theorem; \label{2}
 \item it exists a positive real number $0<L<1$ such that $\|\Delta(\mathbf{x}_{n+1})\|_2 < L\|\Delta(\mathbf{x}_n)\|_2$.\label{3}  \end{enumerate}
From \ref{1}. and \ref{2}. it follows that the fixed point of $G(\mathbf{x})$ in $S$ is unique; from \ref{3}. it follows that $\Delta(\mathbf{x}_n) \rightarrow \mathbf{0}$ and, since $k(\mathbf{x})$ is bounded in $S$, $k(\mathbf{x}_n) \Delta(\mathbf{x}_n) \rightarrow \mathbf{0} $. In conclusion, if $k$ is appropriately chosen, $\{\mathbf{x}_n\}$ tends to a finite limit in $S$ which is the unique fixed point of $G(\mathbf{x})$, that is $\lim_{n\rightarrow \infty} \mathbf x_n= \mathbf x^*$. While in the present section we simply assume the existence of a function $k(\mathbf{x})$ with the desired properties, in Section \ref{stdev} we will show how to compute it explicitly. This is due to the fact that the derivation of $k(\mathbf{x})$ requires the selection of a specific reference risk measure. In our framework, the existence and uniqueness of the risk budgeting portfolio are equivalent to the existence and uniqueness of the fixed point of $G$. Therefore, to apply Kellogg’s theorem, it is necessary that $\Delta$ be continuous.

\begin{remark}(\textbf{Continuity of $\Delta(x)$})
If $\Delta$ is a continuous function and $\{\mathbf{x}_n\}$ is a Cauchy sequence with limit $\mathbf{x}^*$, then $\Delta(\mathbf{x}_n)\rightarrow \mathbf{0}$. Conversely, if $\Delta$ is not continuous, $\mathbf{x}_n\rightarrow \mathbf{x}^*$ does not imply $\Delta(\mathbf{x}_n)\rightarrow \Delta(\mathbf{x}^*)$. Nevertheless, if $\Delta$ is continuous in a neighborhood of $\mathbf{x}^*$, since it definitively belongs to the neighborhood, then local convergence to the fixed point can still be established.  
\end{remark}

\subsection{The algorithm}\label{algo}

In this section, we present the general structure of the algorithm used to compute the risk budgeting portfolio. Although the numerical procedure is implicit in the sequence defined in Subsection \ref{fix_gen}, it is useful to provide a pseudo-code representation and highlight certain technical aspects relevant for implementation. The following presents the general pseudo-code for the practical computation of the risk budgeting portfolio. 
The algorithm operates as follows: starting from a point $\mathbf{x}_n$, the vector $\Delta(\mathbf x_n)$ identifies a direction in the space. The sign of $k_n$ (for simplicity, we use $k_n$ to denote $k(\mathbf x_n)$ for the rest of the paper) determines the orientation along this direction, while the value $|k_n| \|\Delta(\mathbf x_n) \|$ represents the step length. A suitable restriction on $|k_n|$ ensures that the sequence $\mathbf{x}_n$ remains within the simplex $S$.

\begin{algorithm}
\label{algo_gen}
\begin{algorithmic}
\STATE{Define the risk contribution $\textsl{RC}^{\rho}$, the risk budget $\mathbf{x}^b$, a positive parameter $L \in (0,1)$, the tolerance $\textbf{tol}$, the initial point $\mathbf{x}_1 \in S$, the maximum number of iterations $\textbf{maxit}$}
\STATE{Initialize the step counter: $n := 1$}
\WHILE{ $ \| \Delta(\mathbf x_n) \|_2  > \textbf{tol} \quad  \& \quad n < \textbf{maxit}$}
\STATE{Find $k_n: \|\Delta(\mathbf{x}_n + k_n \Delta( \mathbf{x}_n) )\|_2 < L\|\Delta(\mathbf{x}_n)\|_2$}
\IF{ If $\mathbf{x}_n + k_n \Delta \mathbf{x}_n \in S $}
\STATE{ Update $\mathbf{x}_{n+1} = \mathbf{x}_n + k_n \Delta \mathbf{x}_n$}
\ENDIF
\IF{ $\mathbf{x}_n + k_n \Delta \mathbf{x}_n \not\in S $}
\STATE{Calculate $\widetilde k_n=\gamma_n k_n, \ \gamma_n>0 \ : \|\Delta(\mathbf{x}_n + \widetilde k_n \Delta \mathbf{x}_n)\|_2 < L\|\Delta(\mathbf{x}_n)\|_2 \quad \& \quad \mathbf{x}_n + \widetilde k_n \Delta\mathbf{x}_n \in S$\\
Update $\mathbf{x}_{n+1} = \mathbf{x}_n + \widetilde{k}_n \Delta \mathbf{x}_n$}
\ENDIF
\STATE{Update the counter: $n = n + 1$}
\ENDWHILE
\RETURN $\mathbf{x}_{n}$
\end{algorithmic}
\end{algorithm}

The inputs of the algorithm are:
\begin{itemize}
	 \item $\textsl{RC}^{\rho}$, the vector of the marginal risk contributions with respect to the reference risk measure $\rho$;
	 \item $\mathbf{x}^b$, the target vector of risk exposures; 
	 \item $L$, the parameter controlling the convergence rate of the sequence;
	 \item $\textbf{tol}$, the solution tolerance (default value $10^{-6}$);
	 \item $\mathbf{x}_0 \in S$, the initial point for the iteration;
	 \item $\textbf{maxit}$, the maximum number of iterations (default value $1000$)
\end{itemize}

The application of the algorithm requires only the evaluation of the risk contribution vector $\textsl{RC}^{\rho}$ with respect to the selected risk measure $\rho$, and the computation of $k_n$ at each iteration step. In practice, given a constant $0<L<1$, $k_n$ is determined at each step by solving the inequality $\| \Delta(\mathbf{x}_{n+1}) \|_2 < L \| \Delta(\mathbf{x}_n) \|_2$. If $|k_n|$ is sufficiently small to ensure that the next iterate $\mathbf{x}_{n+1}$ remains in the simplex $S$, the step is accepted. Otherwise, $|k_n|$ is appropriately reduced before computing $\mathbf{x}_{n+1}$. 

\begin{remark}
Our method computes $G(x)=x+k(x)*\Delta(x)$ by constructing the sequence $x_{n+1}=G(x_n)$. At each step, we look for $k(x_n)$ such that
$\|\Delta(x_{n+1})\|^2 < L^2\|\Delta(x_n)\|^2, \qquad L <1 \ .$ It follows that the sequence $|\Delta(x_n)|$ tends to zero, and therefore $x_n$ converges to the fixed point of $G$. The formulas resemble those of line search methods, see \cite{sun2006optimization}, which in general allow to compute the unconstrained minimum of a function $f:\mathbb R^N \rightarrow \mathbb R$ by means of the sequence $x_{n+1}=x_n+\alpha_n p_n \ , \alpha_n>0 $ where $p_n$ is a search direction coinciding with a descent direction, that is, such that $\nabla f(x_n)^tp_n<0$, and $\alpha_n$ is the step length. In our case, since we can rewrite the relation in the form $x_{n+1}=x_n + |k(x_n)| \operatorname{sign}(k(x_n)) \Delta(x_n) $, the search direction is $\operatorname{sign}(k(x_n)) \Delta(x_n)$ and the step length is $|k(x_n)| \| \Delta(\mathbf{x}_n) \|_2$. However, despite the similarity, the method is substantially different. Indeed, letting $f(x)=\|\Delta(x)\|^2_2, $ the line search method applied to $f$ computes, at each step, $\alpha_n$ that achieves or approximates the minimum of the function $h(\alpha)= f(x_n+\alpha p_n)$, with $p_n$ a descent direction, whereas our algorithm computes $k(x_n)$ in order to reduce the value $h(k)= \|\Delta(x_n+k\Delta(x_n))\|^2\ ,$ that is, we only require that $h(k) < L h(0)$, with $L<1$. It is essential that $L<1$. Indeed, simply implementing at the $n$-th step a search for $k(x_n)$ such that $f(x_n+k(x_n)\Delta(x_n)) < f(x_n)$ is not an optimal choice, since the algorithm may stagnate; that is, the reductions in the values of the function $f$ may be negligible.
\end{remark}

\section{The case of the standard deviation}\label{stdev}

In this section, we apply our approach by selecting the standard deviation as the reference risk measure. In this case, the vector of marginal risk contributions is given by $\textit{RC}^{\rho}(\mathbf{x}) = \textrm{diag} (\mathbf{x}) V \mathbf{x}$, where $V$ denotes the covariance matrix.\footnote{The covariance matrix is a real, symmetric, positive definite matrix of order $N$.}
Accordingly, by exploiting the expression for the marginal risk contributions, the function $\Delta: \mathbb R^N \rightarrow \ \mathbb R^N$ becomes
$$ \Delta(\mathbf{x})=  \textrm{diag}(\mathbf{x}) V \mathbf{x} - \mathbf{x}^b \textbf{1}^t(\textrm{diag}(\mathbf{x}) V \mathbf{x}), $$
or, equivalently, 
$$ \Delta (\mathbf{x})=  \textrm{diag}(\mathbf{x}) V \mathbf{x} - (\mathbf{x}^t V \mathbf{x}) \mathbf{x}^b. $$

\begin{remark} 
We note that $\textsl{RC}^{\rho}_i(\mathbf{x}) \not \geq 0$ for all $\mathbf{x} \in \mathring{S}$. In fact, it is possible to find portfolios within $\mathring{S}$ (see Example~\ref{es1}) for which the marginal risk contribution of at least one asset is negative. By continuity of $\textsl{RC}^{\rho}$, this also implies the existence of portfolios in $\mathring{S}$ where the marginal risk contribution of a given asset is zero. Recall that $\mathring{S}$ consists of long-only portfolios. Intuitively, when numerous negative correlations are present among assets, the marginal risk contribution of an asset to the total portfolio risk can turn negative. This leads to the counter-intuitive behavior whereby increasing the exposure to that asset may result in a decrease in its 
marginal contribution to total risk (see example \ref{es1}). Regarding the algorithm, in the vast majority of cases, $k_n$ is
 negative, which supports the usual economic intuition: if the risk exposure of an asset exceeds the target, then the corresponding 
allocation should be reduced, and vice versa. In some cases, however, $k_n$ may become positive for some $\mathbf{x}$.\footnote{Based on our experience, such peculiar behaviors tend to arise only when the number of assets $N$ is relatively small. When $N$ is large, we have not observed such phenomena. We conjecture that this effect is related to the structure of the covariance matrix and its positive definiteness.}
\end{remark}

\begin{example}\label{es1}
Let us consider a market with $N=5$ risky assets and the following covariance matrix $V$, which corresponds to the correlation matrix $C$:
\[V = \left[\begin{array}{ccccc} 0.1137  & -0.0289  &  0.0295  &  0.0279  &  0.0437 \\
   -0.0289  &  0.0255 &  -0.0337 &  -0.0156 &  -0.0159 \\
    0.0295  & -0.0337 &   0.1002 &   0.0068 &   0.0427 \\
    0.0279  & -0.0156 &   0.0068 &   0.0281 &   0.0262 \\
    0.0437  & -0.0159 &   0.0427 &   0.0262 &   0.0884 \end{array}\right]
		\]
		\[
		C = \left[\begin{array}{ccccc} 1.0000  & -0.5367  &  0.2764  &  0.4936  &  0.4359 \\
   -0.5367  &  1.0000  & -0.6667 &  -0.5828  & -0.3349\\
    0.2764  & -0.6667  &  1.0000 &   0.1282  &  0.4537\\
    0.4936  & -0.5828  &  0.1282 &   1.0000  &  0.5257\\
    0.4359  & -0.3349  &  0.4537 &   0.5257  &  1.0000 \end{array}\right]
	\]
The eigenvalues of $V$ are $0.0031; 0.0215; 0.0515; 0.0802; 0.1996$. Given the portfolio 
\[
\mathbf{x} = [0.0932; 0.0495; 0.0215; 0.5212; 0.3147]^t
\] 
the corresponding marginal risk contributions are given by
\[
\textit{RC}(\mathbf{x}) = \textrm{diag}(\mathbf{x}) V \mathbf{x}^t = [0.0036; -0.0008; 0.0004; 0.0130; 0.0144]^t.
\] 
Due to the negative correlations between asset $i = 2$ and the other assets, asset $i = 2$ exhibits a negative marginal  risk contribution. If the exposure to asset $i = 2$ is increased of $1\%$ and the exposure to asset $i = 4$ is reduced by the same amount, the updated portfolio and corresponding marginal risk contributions become:
\[
\mathbf{x} = [0.0932; 0.0595; 0.0215; 0.5112; 0.3147]^t
\]  
\[
\textit{RC}(\mathbf{x}) = [0.0036; -0.0009; 0.0004; 0.0125; 0.0142]^t.
\]
This example illustrates a rare behavior in which increasing the exposure to one asset leads to a decrease in its marginal risk contribution to total risk.
Now consider the portfolio
\[
\mathbf{x} = [0.1450; 0.4049; 0.0298; 0.2896; 0.1307]^t
\] 
with marginal risk contributions 
\[
\textit{RC}(\mathbf{x}) = [0.0028; -0.0006; 0.0000; 0.0027; 0.0027]^t.
\]
In this case, a positive exposure to asset $i = 3$ corresponds to a zero marginal risk contribution of that asset to the total risk.   
\end{example}

Given an arbitrary initial point $\mathbf{x}_1 \in \mathring{S}$ and $0<L<1$, the sequences $\{\mathbf{x}_n\}_{n=1,2,\dots}$ and $\{\Delta \mathbf{x}_n\}_{n=1,2,\dots}$ are 
\begin{eqnarray} \label{succ1}
\left \{ \begin{array}{lcl}
\Delta (\mathbf{x}_n) &=& \textrm{diag}(\mathbf{x}_n) V \mathbf{x}_n - (\mathbf{x}_n^t V \mathbf{x}_n) \mathbf{x}^b \\ \\ 
\mathbf{x}_{n+1} &=& \mathbf{x}_n + k_n \Delta (\mathbf{x}_n) 
\end{array} \right . \ ,
\end{eqnarray}
where $k_n$ is computed as follows. At each step $n$, the vectors $D_1$, $D_2$ and $D_3$ are calculated 
\begin{eqnarray*}
D_1&=& \textrm{diag}(\mathbf{x}_n) V \Delta (\mathbf{x}_n) -(\mathbf{x}_n^t V \Delta (\mathbf{x}_n)) \mathbf{x}^b  \ ,\\
D_2&=& \textrm{diag}(\Delta (\mathbf{x}_n) ) V \mathbf{x}_n - ((\Delta (\mathbf{x}_n))^t V  \mathbf{x}_n )  \mathbf{x}^b  \ , \\
D_3&=& \Delta (\Delta (\mathbf{x}_n) )   \ .
\end{eqnarray*}
We propose two alternative criteria for selecting the value of $k_n$.  
\begin{itemize}
\item If the quartic polynomial $Q(k_n)$ has at least one real root, then $k_n$ is chosen such that $Q(k_n) < 0$. 
\begin{eqnarray*}
Q(k_n)&=&k^4 \|D_3\|^2 + 2 k^3 (D_1+D_2)^tD_3  + k^2 (2\Delta (\mathbf{x}_n^t) D_3 + \|D_1+D_2\|^2) \\
&+& 2 k \Delta (\mathbf{x}_n ^t) (D_1+D_2)  + (1-L^2)\|\Delta (\mathbf{x}_n) \|^2 \  . 
\end{eqnarray*}
\item Alternatively, when $Q(k)$ has no real roots, we consider the cubic polynomial $C(k)$,
\begin{eqnarray*}
C(k)&=&k^3 \|D_3\|^2 + 2 k^2 (D_1+D_2)^tD_3  + k (2\Delta (\mathbf{x}_n^t) D_3 + \|D_1+D_2\|^2) \\
&+& 2  \Delta (\mathbf{x}_n^t) (D_1+D_2) \  . 
\end{eqnarray*}
\noindent that always admits at least one real root. Then $k_n$ is chosen to verify $C(k_n) < 0$.
\end{itemize}
As shown in the following theorem, the condition $Q(k_n)<0$ is equivalent to $\|\Delta(\mathbf{x}_{n+1})\|_2^2 \le L^2 \|\Delta(\mathbf{x}_n)||_2^2$ while the condition $C(k_n) < 0$ is equivalent to $\|\Delta(\mathbf{x}_{n+1})\|_2^2 \le \|\Delta(\mathbf{x}_n)\|_2^2$.
The following theorem holds. 
\begin{theorem}\label{lim_delta}
If $k_n$ is calculated as described above and $C(k_n) = 0$ is solved a finite number of times, then 
$$\lim_{n \rightarrow \infty} \|\Delta (\mathbf{x}_n)\|_2=0 \,$$
that is $\lim_{n \rightarrow \infty} \Delta (\mathbf{x}_n)=\mathbf{0}$.
\end{theorem}
\begin{proof}
It holds that 
\begin{eqnarray*}
\Delta (\mathbf{x}_{n+1})&=& \textrm{diag}(\mathbf{x}_{n+1}) V \mathbf{x}_{n+1} - \mathbf{x}^b (\mathbf{x}_{n+1}^t V \mathbf{x}_{n+1}) \\
&=& \textrm{diag}(\mathbf{x}_n +k_n \Delta (\mathbf{x}_n)) V (\mathbf{x}_n + k_n \Delta (\mathbf{x}_n)) \\
&&- \mathbf{x}^b \left ((\mathbf{x}_n +k_n \Delta (\mathbf{x}_n))^t V (\mathbf{x}_n + k_n \Delta (\mathbf{x}_n))\right ) \\
&=&  \textrm{diag}(\mathbf{x}_n ) V \mathbf{x}_n + k_n \textrm{diag}(\mathbf{x}_n ) V \Delta (\mathbf{x}_n) \\
&&+ k_n \textrm{diag}( \Delta (\mathbf{x}_n)) V \mathbf{x}_n  +
k_n^2 \textrm{diag}(\Delta (\mathbf{x}_n)) V (\Delta (\mathbf{x}_n)) \\
&&- \mathbf{x}^b \left (\mathbf{x}_n ^t V \mathbf{x}_n \right )  - 2 k_n \mathbf{x}^b \mathbf{x}_n^t V \Delta (\mathbf{x}_n) -\mathbf{x}^b k_n^2 \left (\Delta (\mathbf{x}_n^t)  V \Delta (\mathbf{x}_n)\right ) \\
&=& \Delta \mathbf{x}_n + k_n (D_1 +D_2) + k_n^2 D_3 
\end{eqnarray*}
Imposing $\|\Delta (\mathbf{x}_{n+1}) \|^2\le L^2  \|\Delta (\mathbf{x}_n)\|^2$,
\begin{eqnarray*}
&&\left ( \Delta (\mathbf{x}_n) + k_n (D_1 +D_2) + k_n^2 D_3\right)^t\left ( \Delta (\mathbf{x}_n) + k_n (D_1 +D_2) + k_n^2 D_3\right ) \\
&&\le L^2  \|\Delta (\mathbf{x}_n)\|^2  
\end{eqnarray*}
then
\begin{eqnarray*}
&&(1-L^2)\|\Delta (\mathbf{x}_n) \|^2 + 2 k_n \Delta (\mathbf{x}_n)^t (D_1+D_2)  + k_n^2 (2\Delta (\mathbf{x}_n)^t D_3 + \|D_1+D_2\|^2) \\
&&+ 2 k_n^3 (D_1+D_2)^tD_3 +k_n^4 \|D_3\|^2 \le 0. 
\end{eqnarray*}
If $Q(k)$ has at least one real root, let $\alpha$ be the root with the smallest absolute value. Since $Q(0) >0$, if $\alpha >0$, then any value of $k_n> \alpha$, smaller than an eventual root larger than $\alpha$, verifies the inequality. Conversely, if $\alpha <0$, a value of $k_n <\alpha$, larger than an eventual root smaller than $\alpha$, verifies the inequality. The value of $k_n$ permits to calculate $x_{n+1}$ such that $\|\Delta (\mathbf{x}_{n+1}) \|_2\le L  \|\Delta (\mathbf{x}_n)\|_2$. 
\smallskip
If $Q(k)$ has non real roots, then we require $\|\Delta (\mathbf{x}_{n+1}) \|^2\le \|\Delta (\mathbf{x}_n)\|^2$, which is equivalent to the following inequality
 \begin{eqnarray*}
&&k_n( 2  \Delta (\mathbf{x}_n )^t (D_1+D_2)  + k_n(2\Delta (\mathbf{x}_n)^t D_3 + \|D_1+D_2\|^2) \\
&&+ 2 k_n^2 (D_1+D_2)^tD_3 +k_n^3 \|D_3\|^2 ) \le 0. 
\end{eqnarray*}
A cubic polynomial always has at least one real root. Therefore, it is possible to calculate $k_n$ such that $\|\Delta (\mathbf{x}_{n+1}) \|^2\le \|\Delta (\mathbf{x}_n)\|^2$.\footnote{When solving the cubic inequality, it is necessary to analyze the sign of $k_n C(k_n)$. Four cases arise depending on the sign of $k_n$ and the behavior of $C(k_n)$.} If the real root of the cubic polynomial is calculated a finite number of times, it holds that
$$ \|\Delta (\mathbf{x}_{n+1}) \|_2\le L^m\|\Delta (\mathbf{x}_n)\|_2,$$
where $m \le n$, with $\lim_{n\rightarrow \infty} m=\infty$. Since $L<1$, the first result holds. Obviously, since the norm is a continuous function, $\lim_{n\rightarrow \infty} \Delta(\mathbf{x}_n) = \mathbf{0}$.
\end{proof}

Now we prove that, $\{\mathbf{x}_n\}$ is a Cauchy sequence when $k_n$ is calculated as described. This result is a special case of  theorem \ref{main}, see subsection \ref{fix_gen}. Compared to the general version, the following theorem differs in two main aspects: first, it exploits the specific functional form of the marginal risk contributions; second, it explicitly incorporates the calculation of $k_n$ within the proof, rather than assuming the existence of $k_n$.
\begin{theorem}
If $k_n$ is calculated as described, $C(k_n) = 0$ is solved a finite number of times and $|k_n| < k$, $\forall n$, (i.e. the sequence $\{k_n\}_{n=1,2,\dots} $ is bounded), then $\{\mathbf{x}_n\}$ is a Cauchy sequence.
\end{theorem}
\begin{proof}
Without loss of generality, we assume $q = p + j$, with $j \in \mathbb{N}$. Then
\begin{eqnarray*}
&&\|\mathbf{x}_{p+j} - \mathbf{x}_p \|_2 =\|\mathbf{x}_{p+j}- \mathbf{x}_{p+1}+\mathbf{x}_{p+1} - \mathbf{x}_p \|_2 \\
&&\le \|\mathbf{x}_{p+j}- \mathbf{x}_{p+1}\|_2+ \|\mathbf{x}_{p+1} - \mathbf{x}_p \|_2 \\
&&\le \|\mathbf{x}_{p+j}- \mathbf{x}_{p+2}\|_2+ \|\mathbf{x}_{p+2} - \mathbf{x}_{p+1} \|_2+ \|\mathbf{x}_{p+1} - \mathbf{x}_p \|_2 \\
&&\le  \dots \le  \sum_{i = 0}^{j-1} \|\mathbf{x}_{p+i+1}-\mathbf{x}_{p+i}\|_2 = \sum_{i = 0}^{j-1} \|k_{p+i}\Delta \mathbf{x}_{p+i}\|_2 \\
&& \le k \sum_{i = 0}^{j-1} \|\Delta (\mathbf{x}_{p+i})\|_2 \ . 
\end{eqnarray*}
By assumption, it exists $L < 1 $ such that $\|\Delta (\mathbf{x}_{n+1})\|_2 \le L \|\Delta (\mathbf{x}_n)\|_2$ is satisfied an infinite number of times since $\|\Delta (\mathbf{x}_{n+1})\|_2 = \|\Delta (\mathbf{x}_n)\|_2$ can happen in a finite number of steps, it holds that
 $$\|\Delta (\mathbf{x}_{n+h})\|_2 \le L^{m_h} \|\Delta (\mathbf{x}_n)\|_2 \quad \text{with} \quad m_h \leq h$$
and we obtain
\begin{eqnarray*}
&&\le k \sum_{i = 0}^{j-1} \|\Delta \mathbf{x}_{p+i}\|_2\le k \sum_{i = 0}^{j-1} L^{m_i} \|\Delta \mathbf{x}_p\|_2 \\
&&< k\|\Delta \mathbf{x}_p\|_2 \sum_{i = 0}^{\infty }  L^{m_i}= k\|\Delta \mathbf{x}_p\|_2 \frac 1 {1-L} \ ,
\end{eqnarray*}
because $ \lim_{n \rightarrow \infty}\sum_{i = 0}^n L^i = \frac{1}{1-L}$.
The limit of the sequence $\{\Delta (\mathbf{x}_n)\}$ is $0$, given $\epsilon$ and $\gamma=\epsilon\frac{1-L}k$, it exists $\overline n$ such that, if $p> \overline n$ then $\|\Delta (\mathbf{x}_p)\|_2< \gamma$. 
Concluding, for each $\epsilon$ it exists $\overline n$ such that for $p> \overline n$ and, consequently, $q>p>\overline n$,  the following holds 
\begin{eqnarray*}
\|\mathbf{x}_{p} - \mathbf{x}_q \|_2 \le k\|\Delta (\mathbf{x}_p)\|_2 \frac 1 {1-L} < k \gamma  \frac 1 {1-L} =\epsilon \ .
\end{eqnarray*}
\end{proof}
Each Cauchy sequence has a finite limit within a complete metric space. Therefore, $\{ x_n\}$ has a finite limit in $\mathbb R^N$.  

\begin{theorem}\label{lallo}
Let $\widetilde{\mathbf{x}} \in \mathbb R^N$ such that $\lim_{n\rightarrow \infty} \mathbf{x}_n=\widetilde{\mathbf{x}}$, then the following holds
 $$\textrm{diag} (\widetilde{\mathbf{x}})V\widetilde{\mathbf{x}}=\mathbf{x}_b (\widetilde{\mathbf{x}}^{t}V\widetilde{\mathbf{x}}) \ . $$
\end{theorem}
\begin{proof}
Thanks to theorem~\ref{lim_delta},  $\lim_{n\rightarrow \infty} \Delta \mathbf{x}_n =\mathbf{0}$. \\
The function $\Delta (\mathbf{x})$ is continuous and $\lim_{n\rightarrow \infty} \mathbf{x}_n=\widetilde{\mathbf{x}}$, then
 $$ \lim_{n\rightarrow \infty} \Delta (\mathbf{x}_n)= \Delta (\widetilde{\mathbf{x}}) =  \textrm{diag} (\widetilde{\mathbf{x}})V\widetilde{\mathbf{x}}-\mathbf{x}_b (\widetilde{\mathbf{x}}V\widetilde{\mathbf{x}})  \ .$$
Concluding,  
$$ \textrm{diag} (\widetilde{\mathbf{x}})V\widetilde{\mathbf{x}}-\mathbf{x}_b (\widetilde{\mathbf{x}}^{t}V\widetilde{\mathbf{x}})  = \textbf{0}$$
from which we obtain the thesis of the theorem.
\end{proof}
From theorem \ref{lallo}, it follows that $\mathbf{x}^*$ is a fixed point of the function $G: \mathbb R^N \rightarrow \mathbb R^N$, defined as $G(\mathbf{x})=\mathbf{x}- k(\mathbf{x}) \Delta (\mathbf{x})$. This result is immediate, since $\Delta (\mathbf{x}^*) = \mathbf{0}$, $k(\mathbf{x})$ is bounded and thus $G(\mathbf{x}^*)=\mathbf{x}^*$. However, as discussed more generally at the end of subsection \ref{fix_gen}, the fixed point may not be unique. 

\subsection{Existence and uniqueness of the fixed point in $S$}
In subsection \ref{exuniq}, we presented general results concerning the existence and uniqueness of the fixed point. When the standard deviation is adopted as the referring risk measure, an interesting alternative to Theorem~\ref{kellogg} becomes available for proving the existence and uniqueness of the fixed point in $S$. Specifically, the functional form of the marginal risk contributions allows the application of general mathematical results developed in the context of rescaling the row and column norms of matrices to improve their condition number (see, among others, \cite{o2003scaling}). Reformulating theorem~1 and Corollary~1 from \cite{o2003scaling} within our framework the following result holds.

\begin{theorem}[\textbf{Oleary}] \label{th:Oleary}
Given a positive definite square real matrix $V$ of order $N$ and a vector $x_b \in \mathbb{R}^N$, it exists a vector
$\mathbf{y} \in \mathbb R^N$ such that 
\begin{equation}\label{eq:OLeary}
\textrm{diag}(\mathbf{y}) V \mathbf{y} = \mathbf{x}_b \ .
\end{equation}
Moreover, equation~\eqref{eq:OLeary} admits $2^N$ solutions, each located in a distinct orthant of $\mathbb{R}^N$.
\end{theorem}
The following corollary is an immediate consequence of the previous result.
\begin{corollary}
Given a positive definite square real matrix $V$ of order $N$ and a vector $\mathbf{x}_b \in S$, it exists a unique fixed point $ \mathbf{\widehat{x}}$ of $G(\mathbf{x})=\mathbf{x}-(\textrm{diag}(\mathbf{x})V\mathbf{x}-\mathbf{x}_b(\mathbf{x}^tV\mathbf{x}) $ in $S$. 
\end{corollary}

\begin{proof}
From theorem~\ref{th:Oleary}, it exists a unique $ \mathbf{\widehat{y}} \in \mathbb R^{N\times N}$ with positive entries such that $\textrm{diag}( \mathbf{\widehat{y}}) V  \mathbf{\widehat{y}} = \mathbf{x}_b $.
Since $\mathbf{x}_b \in S$, then
$$1= \mathbf 1^t \mathbf{x}_b =  \mathbf 1^t(\textrm{diag}( \mathbf{\widehat{y}}) V  \mathbf{\widehat{y}})=  \mathbf{\widehat{y}}^tV  \mathbf{\widehat{y}}$$
and, consequently, 
$$\textrm{diag}( \mathbf{\widehat{y}}) V  \mathbf{\widehat{y}} = 1\cdot \mathbf{x}_b= \left ( \mathbf{\widehat{y}} ^tV \mathbf{\widehat{y}} \right ) \mathbf{x}_b \ .$$
Let 
$$ \mathbf{\widehat{z}}=\frac{ \mathbf{\widehat{y}}}{\sum_{i=1}^N \widehat y_i} =\frac{ \mathbf{\widehat{y}}}{\|  \mathbf{\widehat{y}}\|_1}$$Note that $ \mathbf{\widehat{z}} \in S$ because the entries are positive real numbers with unitary sum.
From 
$$\textrm{diag}( \mathbf{\widehat{y}}) V  \mathbf{\widehat{y}} =  \left ( \mathbf{\widehat{y}}^tV  \mathbf{\widehat{y}} \right ) \mathbf{x}_b$$
 it holds that
$$ \textrm{diag}\left (\frac{ \mathbf{\widehat{y}}}{\| \mathbf{\widehat{y}}\|_1}\right ) V \frac{ \mathbf{\widehat{y}}}{\| \mathbf{\widehat{y}}\|_1}  =  \left (\frac{ \mathbf{\widehat{y}}^t}{\| \mathbf{\widehat{y}}\|_1}V\frac{ \mathbf{\widehat{y}}}{\|   \mathbf{\widehat{y}}\|_1}\right ) \mathbf{x}_b $$
equivalently, 
$$\textrm{diag}( \mathbf{\widehat{z}}) V  \mathbf{\widehat{z}} = \left ( \mathbf{\widehat{z}}^tV \mathbf{\widehat{z}}\right ) \mathbf{x}_b\qquad \Leftrightarrow \qquad \Delta( \mathbf{\widehat{z}})=0. $$
This proves that the root of $\Delta(\mathbf{x})$ is unique. Then the fixed point of $G(\mathbf{x})$ is unique in $S$.
\end{proof}

In general, as previously noted, the sequence ${\mathbf{x}_n}$ does not necessarily remain within the simplex $S$, unless the step size $|k_n|$ is appropriately bounded. The following theorem provides an explicit condition on $|k_n|$ to ensure that the sequence ${\mathbf{x}_n}$ remains in the simplex. This result is particularly relevant for the practical implementation of the numerical procedure.\footnote{Based on our empirical experience, in the vast majority of cases, the sequence ${\mathbf{x}_n}$ converges to the fixed point within $S$ without any explicit restriction on $|k_n|$. However, in a small number of instances, we observed a finite number of iterations falling outside of $S$. In rare situations, when no restriction is imposed on $|k_n|$, the sequence may converge to a solution located in a different orthant. This typically occurs when an alternative fixed point lies close outside the border of the simplex. For this reason, it is useful to enforce the restriction on $|k_n|$ and ensure the sequence remains within $S$.}

\begin{theorem}\label{normrest}
If $\mathbf{x}_n \in S$, then $\mathbf{x}_{n+1}=\mathbf{x}_n+k_n\Delta(\mathbf{x}_n) \in S$ if 
\begin{eqnarray*}
|k_n| \le \left\{\begin{array}{l} \min \left \{ \frac{1-(\mathbf{x}_n)_i}{ |\Delta(\mathbf{x}_n)_i | }  \ | \ k_n\Delta(\mathbf{x}_n)_i >0  \right  \} \\ \\
 \min \left \{ \frac{(\mathbf{x}_n)_i}{ |\Delta(\mathbf{x}_n)_i| }  \ | \ k_n\Delta(\mathbf{x}_n)_i <0  \right  \} \ ,
\end{array}
\right.
\end{eqnarray*}
where $(\mathbf{x}_n)_i$ and $\Delta(\mathbf{x}_n)_i$ are the $i^{th}$-entries of $\mathbf{x}_n$ and $\Delta(\mathbf{x}_n)$ respectively.
\end{theorem}
\begin{proof}
First,
$ \mathbf 1^t \Delta(\mathbf{x}_n)=\mathbf{x}_n^tV\mathbf{x}_n-(\mathbf 1^t \mathbf{x}_b)(\mathbf{x}_n^tV\mathbf{x}_n)=\mathbf{x}_n^tV\mathbf{x}_n-\mathbf{x}_n^tV\mathbf{x}_n=0$; this means that if the entries of $\mathbf{x}_n$ have unitary sum, $\mathbf{x}_{n+1}$ has the same property if $\mathbf 1^t \mathbf{x}_{n+1}=\mathbf 1^t(\mathbf{x}_n+k_n \Delta(\mathbf{x}_n))=1$, 
independently from the value of $k_n$. Therefore, it is necessary to prove that each component of $\mathbf{x}_{n+1}$ is non-negative. For simplicity, we omit the index $n$ and refer to $\Delta$ for $\Delta(\mathbf{x})$. We require 
$ 0 \le \mathbf{x}_i +k \Delta_i \le 1$, i.e. $ -\mathbf{x}_1 \le k \Delta_i \le 1-\mathbf{x}_i$, 
where $\mathbf{x}_i$ and $\Delta_i$ are the $i^{th}$ entries of $\mathbf{x}$ and $\Delta$ respectively. By assumption, $\mathbf{0} < \mathbf{x}_1 < \mathbf{1}$.
Two cases are possible.
\smallskip\noindent
\textbf{case $1$}: $k \Delta_i >0$, then $k \Delta_i =|k \Delta_i |= |k|| \Delta_i |$. The value of $|k|$ needs to satisfy the conditions 
\begin{eqnarray*}
\left\{\begin{array}{rcl} |k|&\ge& -\frac{\mathbf{x}_i}{ |\Delta_i| } \\ \\
|k| & \le & \frac{1-\mathbf{x}_i}{ |\Delta_i| }
\end{array}
\right. 
\end{eqnarray*}
Note that the first inequality is always verified.

\smallskip\noindent
\textbf{case $2$}: $k \Delta_i <0$, then $k \Delta_i =-|k \Delta_i |= -|k|| \Delta_i |$. The value of $|k|$ needs to satisfy the conditions 
\begin{eqnarray*}
\left\{\begin{array}{rcl} -|k|&\ge& -\frac{\mathbf{x}_i}{ |\Delta_i| } \\ \\
-|k| & \le & \frac{1-\mathbf{x}_i}{ |\Delta_i| }
\end{array}
\right. \qquad \Leftrightarrow \qquad  \left\{\begin{array}{rcl} |k|&\le& \frac{\mathbf{x}_i}{ |\Delta_i| } \\ \\
|k| & \ge & \frac{\mathbf{x}_i-1 }{ |\Delta_i| }
\end{array}
\right. 
\end{eqnarray*}
In this case, the second inequality is always verified.

\smallskip\noindent
Concluding, $|k|$ verifies 
\begin{eqnarray*}
|k| \le \left\{\begin{array}{l} \min \left \{ \frac{1-\mathbf{x}_i}{ |\Delta_i| }  \ | \ k\Delta_i >0  \right  \} \\ \\
 \min \left \{ \frac{\mathbf{x}_i}{ |\Delta_i| }  \ | \ k\Delta_i <0  \right  \}
\end{array}
\right.
\end{eqnarray*}
\end{proof}

\section{Empirical results}\label{num}

In this section, we present a numerical application to demonstrate the effectiveness of the proposed approach for computing the risk budgeting portfolio. The reference risk measure employed is the standard deviation. Accordingly, we implement the fixed point (FP) algorithm introduced in general terms in Section~\ref{algo}, specifying the marginal risk contributions and the step size $k_n$ as outlined in Section~\ref{stdev}. The performance of the FP algorithm is compared against three standard methods commonly used in the literature. A brief description of these alternative approaches is provided below.

\begin{itemize}
	\item Optimization 1 (OP1): The risk budgeting portfolio is computed by solving the following optimization problem using the MATLAB function \textsl{fmincon}.
	\begin{eqnarray*}
\begin{array}{lll}
\textrm{Minimize} && \sum_{i = 1}^N \sum_{j = 1}^N \left(\mathbf{x}_i (V\mathbf{x})_i - \mathbf{x}^b_j  \right)^2\\
\textrm{subject to} && \mathbf{x}_i \geq 0\\
&& \sum_{i = 1}^N \mathbf{x}_i = 1. 
\end{array}
\end{eqnarray*}
This optimization problem was originally proposed in \cite{maillard2010properties} for the calculation of the risk parity portfolio with respect to the standard deviation. In our case, we adapt it by modifying the objective function in order to compute the risk budgeting portfolio.

	\item Optimization 2 (OP2): The risk budgeting portfolio is computed by solving the following optimization problem using the MATLAB function \textsl{fmincon}
	\begin{eqnarray*}
\begin{array}{lll}
\textrm{Minimize} && \sqrt{\|\mathbf{y}^t V \mathbf{y} - \mathbf{x}^b \|_2}\\
\textrm{subject to} && \sum_{i = 1}^N \log{\mathbf{y}_i} \geq c\\
&& \mathbf{y}_i \geq 0 
\end{array}
\end{eqnarray*}
where $\log$ denotes the natural logarithm and $c$ is an arbitrary constant. This optimization problem was originally proposed in \cite{maillard2010properties} as a useful alternative to OP1 for computing the risk parity portfolio. In our case, we slightly modified the objective function to generalize the problem and derive the risk budgeting portfolio. As highlighted in the introduction, this optimization approach is also employed in \cite{cetingoz2024risk} to establish sufficient conditions for the existence and uniqueness of the risk budgeting portfolio and widely used in the literature.

  \item Non-linear system (NLS): the condition $\frac{\textsl{RC}^{\rho}(\mathbf{x}^*)}{\sum_{i = 1}^N \textsl{RC}^{\rho}_i(\mathbf{x})} = \mathbf{x}^b$ which characterizes the risk budgeting portfolio for a general risk measure $\rho$, defines a nonlinear system of equations. In principle, it is always possible to attempt solving this system numerically to obtain the risk budgeting portfolio. For this purpose, we use the MATLAB function \textsl{lsqnonlin}.
\end{itemize}

All the MATLAB built-in functions used in the application are implemented with their default settings. In particular, the maximum number of iterations is determined by the standard configuration. Increasing the number of iterations may lead to a more accurate solution; however, this typically comes at the cost of increased computational time. The different algorithms are compared based on two criteria: the computational time required to reach the solution and the distance from the effective solution. If $\overline{\mathbf{x}}$ denotes the solution obtained by a given algorithm, the distance to the true risk budgeting portfolio is measured by the 
$\|\textsl{RC}^{\rho}(\overline{x})/\boldsymbol{1}^t \textsl{RC}^{\rho}(\overline{x}) - \mathbf{x}^b \|_2$. This metric evaluates how closely the realized risk contributions match the target risk allocation $\mathbf{x}^b$.

Since computational time may depend on the portfolio size $N$, we compare the performance of the algorithms for various values: $N = 5, 10, 50, 100, 200$. For each value of $N$, we generate $T = 1000$ random experiments. In each experiment, a covariance matrix $V$ is sampled, along with two vectors $\mathbf{x}_1$ and $\mathbf{x}^b$ in $S$, representing the initial point and the target risk budget, respectively. Both vectors are uniformly sampled from the simplex $S$ to stress the numerical efficiency of the algorithms, particularly when one or both are near the boundary of the simplex, which can lead to increased computational complexity. To further assess the robustness of the algorithms, some covariance matrices are chosen to be numerically close to a rank-deficient matrix. For each triplet $(V, \mathbf{x}_1, \mathbf{x}^b)$, the four algorithms are applied, and both the computational time and the distance from the target solution are recorded. Table~\ref{tab1} summarizes the results. For each algorithm and each portfolio size, we report the average convergence time and the mean distance from the true risk budgeting solution, as defined above. The key observation from Table~\ref{tab1} is that the fixed-point (FP) algorithm consistently outperforms the other methods, both in terms of speed and accuracy, across all tested scenarios.

\begin{table}[!h]
  \centering
  \caption{Comparison of average computational time (in seconds) and solution accuracy across different algorithms (FP, OP1, OP2, NLS) and portfolio sizes $(N = 5, 10, 50, 100, 200)$.}
    \begin{tabular}{|r|l|r|r|r|r|}
		\hline
          &       & FP $(L = 0.5)$ & OP1 & NLS & OP2 \\
					\hline
					\hline
    \multicolumn{1}{l}{N = 5} & time(s)  & 0.0007 & 0.0140 & 0.0061 & 0.0158 \\
          & accuracy & 0.0002 & 0.0897 & 0.1527 & 0.3633 \\
					\hline
    \multicolumn{1}{l}{N = 10} & time(s)  & 0.0017 & 0.0228 & 0.0069 & 0.0313 \\
          & accuracy & 0.0004 & 0.0467 & 0.1506 & 0.2916 \\
					\hline
    \multicolumn{1}{l}{N = 50} & time(s)  & 0.0047 & 0.1267 & 0.0363 & 0.0836 \\
          & accuracy & 0.0012 & 0.0397 & 0.0896 & 0.2058 \\
					\hline
    \multicolumn{1}{l}{N = 100} & time(s)  & 0.0100 & 0.2612 & 0.1959 & 0.1089 \\
          & accuracy & 0.0023 & 0.0488 & 0.0626 & 0.1368 \\
					\hline
    \multicolumn{1}{l}{N = 200} & time(s)  & 0.0383 & 1.4118 & 3.6511 & 0.2711 \\
          & accuracy & 0.0040 & 0.0556 & 0.0458 & 0.2877 \\
					\hline
    \end{tabular}%
  \label{tab1}%
\end{table}%

As previously highlighted, the accuracy of the solutions obtained via OP1, OP2, and NLS can be improved by increasing the maximum number of iterations. However, this enhancement would come at the cost of significantly longer convergence times, further penalizing the efficiency of these procedures. It is worth noting that, in addition to outperforming OP1 and OP2, the NLS algorithm demonstrates competitive performance. As expected, computational time increases with the portfolio size $N$ across all algorithms. Nevertheless, the FP algorithm exhibits the greatest resilience to dimensionality, with its performance being least affected by the number of assets. Furthermore, portfolio size does not appear to significantly influence the accuracy of the solution for any of the methods.

We now turn our attention to the FP algorithm, focusing on an application designed to investigate the impact of the parameter $L$. This parameter is externally specified and, in theory, governs the convergence rate—and thus the computational time—of the algorithm. For each triad $(V, \mathbf{x}_1, \mathbf{x}^b)$, we run the FP algorithm while varying the value of $L$, recording both the average computational time and the accuracy of the solution. The outcomes are reported in Table~\ref{tab2}. An interesting observation emerges: the smallest value of $L$ does not always correspond to the fastest convergence. This phenomenon is straightforward to interpret. When $L$ is too small, the algorithm frequently encounters iterations for which the associated polynomial of degree 4 has no real solution. As a result, several steps of the algorithm fail to advance toward the solution, thereby slowing the overall convergence. In contrast, when $L$ approaches 1, although each step length is relatively small, it reliably moves closer to the solution, enhancing overall convergence efficiency.

\begin{table}[!h]
  \centering
  \caption{Comparison of the FP algorithm for different values of the parameter $L = 0.3, 0.5, 0.7, 0.9, 0.95, 0.99$.}
    \begin{tabular}{|l|l|r|r|r|r|r|r|}
		\hline
          &       & FP & FP & FP & FP & FP & FP \\
					&   L    & 0.3 & 0.5 & 0.7 & 0.9 & 0.95 & 0.99 \\
					\hline
					\hline
    N = 5 & time(s)  & 0.0009 & 0.0007 & 0.0009 & 0.0024 & 0.0050 & 0.0234 \\
          & accuracy & 0.0006 & 0.0002 & 0.0001 & 0.0002 & 0.0003 & 0.0002 \\
					\hline
    N = 10 & time(s)  & 0.0014 & 0.0017 & 0.0013 & 0.0033 & 0.0054 & 0.0255 \\
          & accuracy & 0.0003 & 0.0004 & 0.0004 & 0.0003 & 0.0004 & 0.0004 \\
					\hline
    N = 50 & time(s)  & 0.0075 & 0.0047 & 0.0052 & 0.0108 & 0.0204 & 0.1903 \\
          & accuracy & 0.0010 & 0.0012 & 0.0013 & 0.0014 & 0.0016 & 0.0018 \\
					\hline
    N = 100 & time(s)  & 0.0207 & 0.0100 & 0.0103 & 0.0306 & 0.0464 & 0.2280 \\
          & accuracy & 0.0020 & 0.0023 & 0.0027 & 0.0029 & 0.0030 & 0.0030 \\
					\hline
    N = 200 & time(s)  & 0.0747 & 0.0383 & 0.0470 & 0.1183 & 0.2803 & 1.2434 \\
          & accuracy & 0.0029 & 0.0040 & 0.0051 & 0.0057 & 0.0058 & 0.0059 \\
					\hline
    \end{tabular}%
  \label{tab2}%
\end{table}%
Considering the good performance of the FP algorithm relative to alternative procedures, along with its low computational cost per iteration, we recommend selecting values of $L < 1$ close to 1 for practical implementation.

\section{Conclusions}
\label{sec:conclusions}
This paper proposes an alternative approach to risk budgeting, different from the standard optimization-based methods commonly studied in the literature. The advantages of this perspective are manifold. First, it allows for the analysis of existence and uniqueness of the risk budgeting portfolio for a general risk measure by formulating the problem as the existence and uniqueness of a fixed point of a function. Second, the computation of the risk budgeting portfolio does not rely on solving an auxiliary optimization problem that often lacks clear economic interpretation. Additionally, the optimization-based techniques—used in standard approaches—can become computationally demanding, particularly as the number of assets increases. In contrast, the proposed method naturally leads to an algorithm for calculating the risk budgeting allocation, which is characterized by fast convergence and low computational cost, thanks to the simplicity of each iteration. Importantly, the computational burden does not scale significantly with portfolio size. Future research will explore the broader application of this methodology to various risk measures, enabling a detailed comparison with the optimization-based techniques typically employed in practice. From a theoretical standpoint, further work will aim to refine the sufficient conditions for the existence and uniqueness of the risk budgeting portfolio, thereby expanding the class of risk measures compatible with the risk budgeting framework.

\bibliographystyle{abbrvnat}
\bibliography{references}


\end{document}